%% file: paper.tex
\documentclass{llncs}
\usepackage{makeidx}  
\usepackage[utf8]{inputenc}
\usepackage{amsfonts, amsopn}
\usepackage{graphicx}
\usepackage{xspace}							
\usepackage{multirow}						

\newcommand{\lcdj}[1]{
	\lambda'_{2,1}(#1)
}
\newcommand{\lcdjf}[0]{
	\lambda'_{2,1}
}
\newcommand{\fcdjf}[0]{
	f'_{2,1}
}
\newcommand{\fcdj}[1]{
	f'_{2,1}(#1)
}
\newcommand{\N}[0]{
	\mathbb{N}
}

\DeclareMathOperator*{\LL}{L}
\newcommand{\Lprb}[0]{
	\unskip	$\LL_{2,1}$\xspace
}
\newcommand{\Lpr}[0]{
	\unskip	$\LL'_{2,1}$\xspace
}
\newcommand{\Lproblem}[0]{
	\unskip	{\sc Distance Edge Labeling}\xspace
}
\newcommand{\Lvproblem}[0]{
	\unskip	{\sc Distance Vertex Labeling}\xspace
}
\newcommand{\MNproblem}[0]{
	\unskip	{\sc Monotone Not All Equal 3-SAT}\xspace
}
\newcommand{\MNpr}[0]{
	\unskip	{\sc MNAE-3-SAT}\xspace
}

\newcommand{\pprob}[5]{
\begin{definition}[#1]\label{#5}
\newline
\begin{center}
\begin{tabular} {|ll|}
	\hline
	{\bf Input:\enspace}&{\parbox[t]{27em}{#2}}\\
    {\bf Parameter:\enspace}&{\parbox[t]{27em}{#3}}\\
	{\bf Question:\enspace}&\parbox[t]{27em}{#4}\\
	\hline
\end{tabular}
\end{center}
\end{definition}
}

\newcommand{\prob}[4]{
\begin{definition}[#1]\label{#4}
\newline
\begin{center}
\begin{tabular} {|ll|}
	\hline
	{\bf Input:\enspace}&{\parbox[t]{27em}{#2}}\\
	{\bf Question:\enspace}&\parbox[t]{27em}{#3}\\
	\hline
\end{tabular}
\end{center}
\end{definition}
}
\newcommand{\img}[2]{
	\begin{center}
		\includegraphics[scale=#2]{img/#1.eps}
	\end{center}
}
\newcommand{\imgw}[2]{
	\begin{center}
		\includegraphics[width=#2\textwidth]{img/#1.eps}
	\end{center}
}

\newcommand{\NP}{{$\mathsf{NP}$}}
\newcommand{\pNP}{{$\mathsf{paraNP}$}}
\newcommand{\PTIME}{{$\mathsf{P}$}}
\newcommand{\El}[0]{
\mathbb{E}
}
\newcommand{\Ol}[0]{
\mathbb{O}
}

\newtheorem{observation}[theorem]{Observation}

\begin{document}
\frontmatter          
\pagestyle{headings}  
\addtocmark{Distance labeling} 
%
%
%
%
%
\title{Computational complexity of distance edge labeling\thanks{Research was supported by the project Kontakt LH12095, by the project SVV-2015-260223 and by the project CE-ITI P202/12/G061 of GAČR.
}}
\titlerunning{Distatnce edge labeling}  
\author{Dušan Knop\thanks{Author was supported by the project GAUK 1784214.}\inst{1} \and Tomáš Masařík\inst{1}} 
\authorrunning{Dušan Knop and Tomáš Masařík} 
%
\tocauthor{Dušan Knop, Tomáš Masařík}
\institute{Department of Applied Mathematics, Faculty of Mathematics and Physics, Charles University, Prague, Czech Republic,\\
\email{\{knop,masarik\}@kam.mff.cuni.cz},
}
\maketitle              

\begin{abstract}
	\input{src/abstract}
\end{abstract}
%
\input{src/intro}
\input{src/prelim}
\input{src/poly}
\input{src/NP}

\bibliographystyle{splncs}
\bibliography{paper}

\appendix
\input{src/appendix}

\end{document}

%% file: src/abstract.tex
The problem of \Lproblem is a variant of \Lvproblem (also known as \Lprb labeling) that has been studied for more than twenty years and has many applications, such as frequency assignment.

The \Lproblem problem asks whether the edges of a given graph can be labeled such that the labels of adjacent edges differ by at least two and the labels of edges at distance two differ by at least one. Labels are chosen from the set $\{0,1,\dots,\lambda\}$ for $\lambda$ fixed.

We present a full classification of its computational complexity---a dichotomy between the polynomially solvable cases and the remaining cases which are \NP-complete.
We characterise graphs with $\lambda\leq 4$ which leads to a polynomial-time algorithm recognizing the class and we show \NP-completeness for $\lambda\geq5$ by several reductions from \MNproblem. 


\keywords{Computational complexity, distance labeling, line graphs}

%% file: src/intro.tex
\section{Introduction}
We study the computational complexity of the distance edge-labeling problem.
This problem belongs to a wider class of problems that generalize the graph coloring problem.
The task is to assign a set of colors to each vertex, such that whenever two vertices are adjacent, their colors differ from each other. For a survey about this famous graph problem and related algorithms, see \cite{colSurvey}.


We are interested in the so-called \emph{distance labeling}. In this generalization of the former problem the condition enforcing different colors is extended and takes into account also the second neighborhood of a vertex (or an edge). The second neighborhood is the set of vertices (or edges) at distance at most 2. For a survey about distance labelings, we refer to the article by Tiziana Calamoneri~\cite{survey}, as well as her online survey~\cite{onlineSurvey}.

Graph distance labeling has been first studied by Griggs and Yeh~\cite{GriggsLabeling,YehLabeling} in 1992. The problem has many applications, the most important one being frequency assignment~\cite{channel}.
The complexity of \Lprb labeling for a fixed parameter $\lambda$ has been established in~\cite{fiala}. They show a dichotomy between polynomial cases for $\lambda\leq3$ and \NP-complete cases for $\lambda\geq4$.

Moreover, for the usual graph coloring problem there is a theorem of Vizing~\cite{vizing}, which states that for the edge-coloring number $\chi'(G)$ it holds that $\Delta\le\chi'\le\Delta + 1,$ where $\Delta$ is the maximum degree of the graph. For \Lprb labeling there is a general bound due to Havet et al.~\cite{havet-upper}, namely $\lambda\le \Delta^2,$ for $\Delta \ge 79.$

Before we proceed to the formal definition of the corresponding decision problem, we give several definitions of a labeling function of a graph and of the minimal distance edge-labeling number. 
Note that the distance edge-labeling is equivalent to the distance vertex labeling of the associated line-graphs.
A \emph{line-graph} $L(G)$ is a graph derived from another graph $G$ such that vertices of $L(G)$ are edges of $G$ and two vertices $a,\ b$ of $L(G)$ are connected by an edge whenever $a,\ b$ (as edges of $G$) are adjacent.
We define the \emph{distance} between edges of a graph as their distance in the corresponding line-graph.

\begin{definition}[Edge-labeling function]
	Let $G(V,E)$ be a graph. A function $\fcdjf :E\rightarrow \N$ is an {\em edge-labeling}, if it satisfies:
\begin{itemize}
\item $|\fcdj{e}-\fcdj{e'}|\geq2$ for neighboring edges (i.e. those in the distance one),
\item $|\fcdj{e}-\fcdj{e'}|\geq 1$ for edges at distance two.
\end{itemize}
\end{definition}

As usual, we are interested in a labeling that minimizes the number of labels used by a feasible labeling.

\begin{definition}[Minimum distance edge-labeling]
Let $G$ be a graph and $\fcdjf$ an edge-labeling function, we define the graph parameter $\lcdjf$ as:
$$
\lcdj{G}:=\min_{\fcdjf}{\max_{e\in E}{\fcdj{e}}}.
$$
\end{definition}

The size of the range of a (not necessarily optimal) edge-labeling function $\fcdjf$ is called the \emph{span}.

\pprob{\Lproblem problem (also known as \Lpr) }{a graph $G$}{$\lambda\in\N$}{Is $\lcdj{G}\leq\lambda?$}{def:lcdjl}

\subsection{Our results}
Our main result is the following theorem about the dichotomy of the \Lproblem problem.

\begin{theorem}[Dichotomy of distance edge-labeling]\label{thm:main}
The problem \Lpr is polynomial-time solvable if and only if $\lambda\leq4$. Otherwise it is \NP-complete.
\end{theorem}

We derive Theorem~\ref{thm:main} as a combination of Theorem~\ref{thm:poly} that describes all graphs with $\lcdjf\le 4$ and Theorem~\ref{thm:np} presenting the \NP-completeness result. Note that our Theorem~\ref{thm:np} also extends to the following inapproximability result:

\begin{corollary}
The \Lproblem problem cannot be approximated within a factor of $6/5-\varepsilon$, unless \PTIME = \NP.
\end{corollary}

Moreover, according to~\cite{flum-grohe}, the proof implies that the \Lproblem  is \pNP-hard while parameterized by its span. 

%% file: src/prelim.tex
\subsection{Preliminaries}
We state several basic and well-known observations with the connection to Definition~\ref{def:lcdjl}, as well as some notation used in this paper.


For further standard notation in graph theory, we refer to the monograph~\cite{diestel2006graph}.

The first observation gives a trivial lower-bound on $\lcdj{G}$.

\begin{observation}[Max-degree lower-bound]\label{thm:deltaLowerBound}
	Let $G$ be a graph and let $\Delta$ be its maximum degree. Then $\lcdj{G}\geq 2\cdot(\Delta - 1).$
\end{observation}

Note that this observation gives also an upper bound on the max-degree of a graph $G$ with $\lcdj{G}\le\lambda$ for a given $\lambda\in\N.$

\begin{observation}[The symmetry of distance labeling]
Let $G$ be a graph, $f:E\rightarrow\N$ be a (not necessarily optimal) labeling function and $\lambda$ be the span of $f$.
Then also the function $f'(e) = \lambda - f(e)$ is a valid labeling function of span $\lambda.$
\end{observation}

We call such a derived labeling of the edges of a graph a {\em $\lambda$-inversion}.

%% file: src/poly.tex
\section{Polynomial cases}
In this section we give a full description of graphs admitting a labeling with small number of labels, in particular graphs $G$ with $\lcdj{G}\leq 4$. Moreover, these graphs can be recognized in polynomial time. This leads to Theorem~\ref{thm:poly}, which is the main result of this section.

For the ease of presentation we split the proof and statement of the Theorem~\ref{thm:poly} into several lemmas, each for a particular value of $\lcdj{G}$.

\begin{theorem}[Polynomial cases of distance edge-labeling.]\label{thm:poly}
For any graph $G$ and for $\lambda = 0,1,2,3,4$ the \Lproblem problem $\lcdj{G} = \lambda$ (or $\lcdj{G} \le\lambda$)
can be solved in polynomial time. Moreover, it is possible to compute such a labeling in polynomial time.
\end{theorem}

Without loss of generality we deal with connected simple undirected graphs.

First observe, that for $\lambda < 4$ the graph cannot contain a vertex of degree $3$. We use $P_i$ as a symbol for the \emph{path} on $i$ vertices.

\begin{lemma}[Graphs with $\mathbf{\lcdj{G}\leq3}$]
	\begin{itemize}
		\item The only graphs with $\lcdj{G} = 0$ are $P_1$ or $P_2$.
		\item There is no graph with $\lcdj{G} = 1$. 
		\item The only graph with $\lcdj{G} = 2$ is $P_3$.
		\item Finally, graphs with $\lcdj{G} = 3$ are $P_4$ and $P_5$.
	\end{itemize}
\end{lemma}

When $\lambda = 4$, the graph may contain vertices with degree $3$. We call a vertex {\em hairy} if it is of degree $3$ and at least one of its neighbors is of degree $1$. 
We call this degree one vertex, together with the connecting edge, {\em pendant}. Note that any vertex of degree $3$ in a graph $G$ satisfying $\lambda(G) = 4$ cannot have all its neighbors of degree $2$ or greater. It is easy to see that there is no labeling of span $4$ of such a graph. We say that two hairy vertices are {\em consecutive}, if there is no other hairy vertex on a path between them or if there is the only hairy vertex on a cycle. In this particular case the vertex is consecutive to itself.

For the purpose of the following lemmas, we say that a graph is a \emph{generalized cycle} if it is a cycle with several (possibly $0$) pendant edges. 
We say that a graph is a \emph{generalized path} if it is a path with several (possibly $0$)pendant edges.
All observations made in the last paragraphs imply the following lemma:

\begin{lemma}
Let $G$ be a graph satisfying $\lcdj{G}\le 4$, then $G$ is either a generalized path or a generalized cycle.
\end{lemma}

On the contrary not every generalized cycle or path has $\lcdjf \leq 4$. The following lemmas state all the conditions for a generalized cycle or path to satisfy $\lcdjf \leq 4.$

\paragraph{Notation in the proofs} Both proofs are done by a case analysis. For generalized cycles and paths the idea is to label path or cycle while there is the possibility to label all the pendant edges. To do so, we use sequences of numbers representing labels on edges.
Note that it follows from Observation~\ref{thm:deltaLowerBound} that only numbers $0,2,4$ can occur around a hairy vertex and any pendant vertex must get label $2$.
For labelings we use sequences of numbers describing labels of consecutive edges and a symbol $"|"$ for a hairy vertex---so there is a pendant edge on a vertex with label $2.$
This gives us immediately the following observation.

\imgw{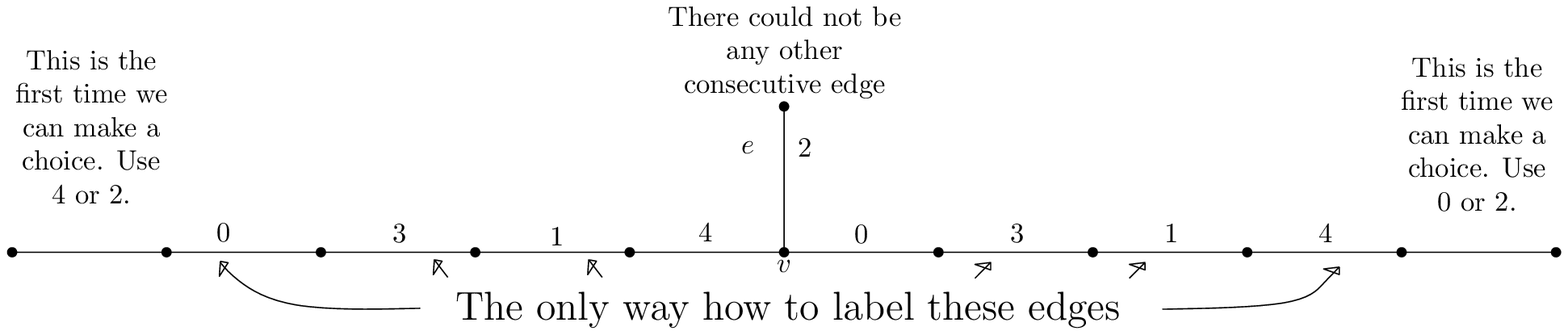}{1}

\begin{observation}[The labeling of a hairy vertex and its neighborhood]\label{obs:hv}
The neighborhood of a hairy vertex can be labeled only by a sequence $0314|0314$ or its $\lambda$-inversion $4130|4130.$
\end{observation}

\begin{lemma}\label{lem:gp}
Let $G=(V,E)$ be a generalized path. Let $W$ be the set of all hairy vertices.
Then $\lcdj{G} \leq 4$ if and only if for every consecutive pair $u,v\in W$ their distance $d = d_G(u,v)$
is either $4$, or at least $8$.
\end{lemma}
\begin{proof}
We need to show that each sequence can be correctly labeled or that it is impossible to label it at all.

The easier fact is the existence of correct labelings.
Sequences $|0314|(d=4), |031420314|(d=9),|0314204130|(d=10), |03140240314|(d=11)$ can be extended by a sequence $0314$ at the beginning to get sequences of length at least $8.$

Now we have to show that there are no valid sequences of length $1,2,3,5,6,7.$
Observation~\ref{obs:hv} banns immediately sequences of length $1,2,3.$
Furthermore, the same observation also implies that there is no chance to overlap two sequences which is necessary to get lengths $5,6$ or $7$.\qed
\end{proof}

\begin{lemma}\label{lem:gc}
Let $G=(V,E)$ be a generalized cycle. Let $W$ be the set of all hairy vertices.
Then $\lcdj{G} \leq 4$ if and only if for every consecutive pair $u,v\in W$ their  distance $d = d_G(u,v)$ fulfills one of the following:
\begin{itemize}
  \item $d = 4,8,9$ or $d \ge 11$,
  \item if there exists a consecutive pair with $d = 10$, then there is even number of such consecutive pairs, or there exists a consecutive pair with $d = 13,14,16$ or greater.
\end{itemize}
\end{lemma}

Firstly it is easy to observe that cycles of any length without a hairy vertex can be labeled correctly.

The proof of the first part is similar to Lemma~\ref{lem:gp}, except for the sequences of length $10$. Because such a sequence cannot be connected via hairy vertex to any sequence presented in the proof of Lemma~\ref{lem:gp}, unless we use a $\lambda$-inversion of some of them. So in the proof of the second part we need to show two things:
\begin{itemize}
	\item The only labeling of a sequence of length $10$ is the one already presented.
	\item The sequences of length less than or equal to $12$ and $15$ do not have a labeling that starts and ends by the label $0$, while sequences of all other possible lengths admit such a labeling.
\end{itemize}

These arguments are proved by a case analysis that is postponed to the appendix.

%% file: src/NP.tex
\section{\NP-complete cases}

\begin{theorem}\label{thm:np}
The problem \Lproblem is \NP-complete for every fixed $\lambda \geq 5$.
\end{theorem}

The proof of the hardness result is done for every $\lambda\geq5$. However as 
there is a natural difference between odd and even $\lambda$, the proof is divided according to the parity of $\lambda$ to two basic general cases. The proof of the even (odd) part is contained in Subsection~\ref{sec:even} and \ref{sec:odd} respectively.

Furthermore, as the gadgets developed to carry the labeling does not work for small cases, we have to exclude the borderline values $\lambda = 5,6,7$ from the general proof.
Due to space limitations, we move these proofs to the appendix.

Our basic reduction tool is the \MNproblem problem which all cases are reduced from. We say a formula $\varphi$ is a \emph{3-MCNF (monotone conjunctive normal form)} if it is a conjunction of clauses with exactly 3 logical variables without negations.

\prob{\MNproblem problem (also known as \MNpr)}{A 3-MCNF formula $\varphi.$}{Is it possible to find an assignment such that each clause has at least one literal set to true and at least one literal set to false?}{def:mnae3sat}

This problem is a specialized version of {\sc NAE-3-SAT}, which was shown to be \NP-complete by Schaefer~\cite{schaefer} by a more general argument about CSP's. We can find \MNpr in the list of \NP-complete problems in the monograph of Garey and Johnson~\cite{Garey}.

\paragraph{The reduction procedure}
For a 3-MCNF formula $\varphi$ and positive integer $\lambda\ge 5$ we show how to build a graph $G_\varphi^\lambda$. We will ensure that $\lcdj{G_\varphi^\lambda}\le\lambda$ if and only if the answer to the question of \MNpr problem is "YES".
In our proofs the main focus is to prove the correspondence between a satisfying assignment to the variables of $\varphi$ and the $\lambda$-labeling of the graph $G_\varphi^\lambda$. We call this the {\em correctness of a gadget}.

\begin{definition}[Odd and Even sets]
For any $\lambda\in\N$ we define two subsets of the set $\{0,\ldots,\lambda\}.$
The {\em odd subset} $\Ol = \{l\in\N\colon l\le\lambda, l\textrm{ odd}\}$ and the {\em even subset} $\El = \{l\in\N\colon l\le\lambda, l\textrm{ even}\}.$
\end{definition}

\begin{example}
Take $\lambda = 10$ (even). Now according to Observation~\ref{thm:deltaLowerBound}, the maximum possible degree of a vertex in a graph admitting a distance labeling with $\lambda$ labels is $6.$ Moreover, only labels from the set $\El$ can appear on edges incident with such a vertex.
\end{example}

\subsection{Basic lemmas}
We state here some auxiliary lemmas that are used in our reductions.

\begin{lemma}[Labeling of edges incident to a maximum degree vertex]\label{lem:max}
	Let $\lambda\in\N$, let $G$ be a graph with $\lcdj{G}\leq\lambda$ and its maximum degree vertex $v$.

	Then:
	\begin{description}
		\item[\textbf{even $\mathbf{\lambda}$}:]  If $deg(v)=\frac{\lambda}{2}+1$ then vertex $v$ has its incident edges labeled by labels from the set $\El$. 
		\item[\textbf{odd $\mathbf{\lambda}$}:] If $deg(v)=\frac{\lambda + 1}{2}$ then a vertex $v$ has its incident edges labeled by labels from the one of the sets: $\Ol$, $\Ol\setminus\{1\}\cup\{0\}$, $\El$ or $\El\setminus\{\lambda-1\}\cup\{\lambda\}$.
	\end{description}
\end{lemma}

\begin{lemma}[Adjacent vertices with maximum degree, even span]\label{lem:even}
	Let $\lambda\in\N$, $\lambda$ even and let $G=(V,E)$ be a graph with $\lcdj{G}\leq\lambda$. Take two neighboring vertices $u,v\in V$ such that $deg(u)=\frac{\lambda}{2}+1$, $deg(v)=\frac{\lambda}{2}$ and $\{u,v\}\in E$.

	Then there are only two possibilities:
	\begin{itemize}
		\item The edge $\{u,v\}$ is labeled by $0$, all the edges incident to $u$ are labeled by the elements from the set $\El\setminus\{0\}$ and finally all the edges incident to $v$ are labeled by the elements from the set $\Ol\setminus\{1\}$.
		\item The edge $\{u,v\}$ is labeled by $\lambda$, all the edges incident to $u$ are labeled by the elements from the set $\El\setminus\{\lambda\}$ and finally all the edges incident to $v$ are labeled by the elements from the set $\Ol\setminus\{\lambda-1\}$.
	\end{itemize}
\end{lemma}

\imgw{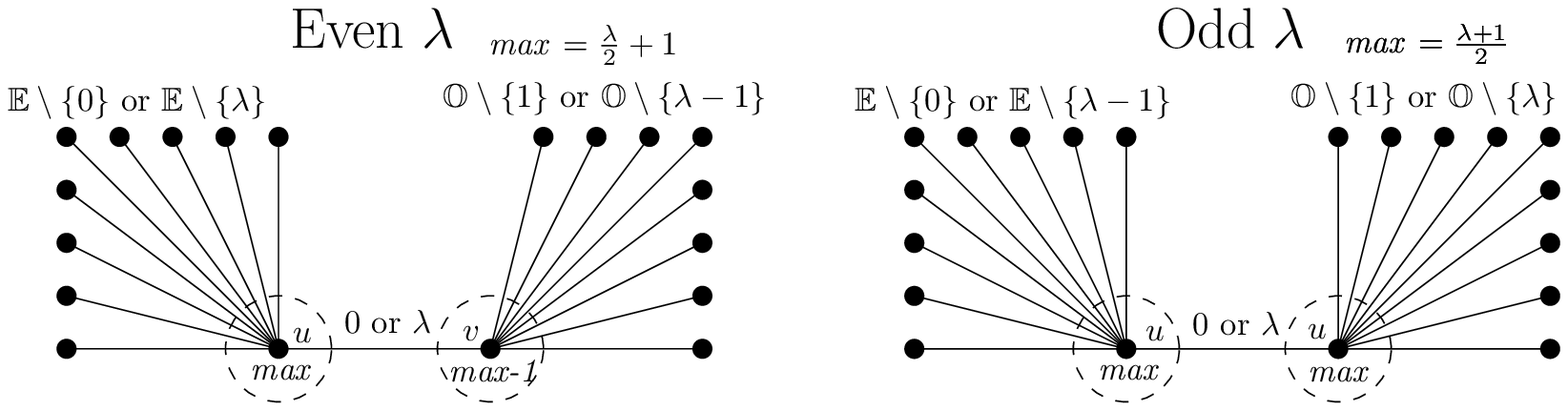}{1}

\begin{lemma}[Adjacent vertices with maximum degree, odd span]\label{lem:odd}
	Let $\lambda\in\N$, $\lambda$ odd and let $G=(V,E)$ be a graph with $\lcdj{G}\leq\lambda$. Take two neighboring vertices $u,v\in V$ such that $deg(u)=deg(v)=\frac{\lambda+1}{2}$.

	Then there are only two possibilities:
	\begin{itemize}
		\item The edge $\{u,v\}$ is labeled by $0$, all the edges incident to $u$ are labeled by the elements from the set $\El\setminus\{0\}$ and finally all the edges incident to $v$ are labeled by the elements from the set $\Ol\setminus\{1\}$.
		\item The edge $\{u,v\}$ is labeled by $\lambda$, all the edges incident to $u$ are labeled by the elements from the set $\El\setminus\{\lambda-1\}$ and finally all the edges incident to $v$ are labeled by the elements from the set $\Ol\setminus\{\lambda\}$.
	\end{itemize}
\end{lemma}

Proof of both lemmas above is an easy application of Lemma~\ref{lem:max}.

\paragraph{Notation in gadgets}
We further use $\max$ as the number for the maximum degree in graph $G$ with $\lcdj{G}\leq \lambda.$
We also use directed edges in gadget graphs. An outgoing edge represents an {\em output}, while an ingoing edge represents an {\em input} to the gadget.
We build all the gadgets so that the labels on output edges can take only several values.

\imgw{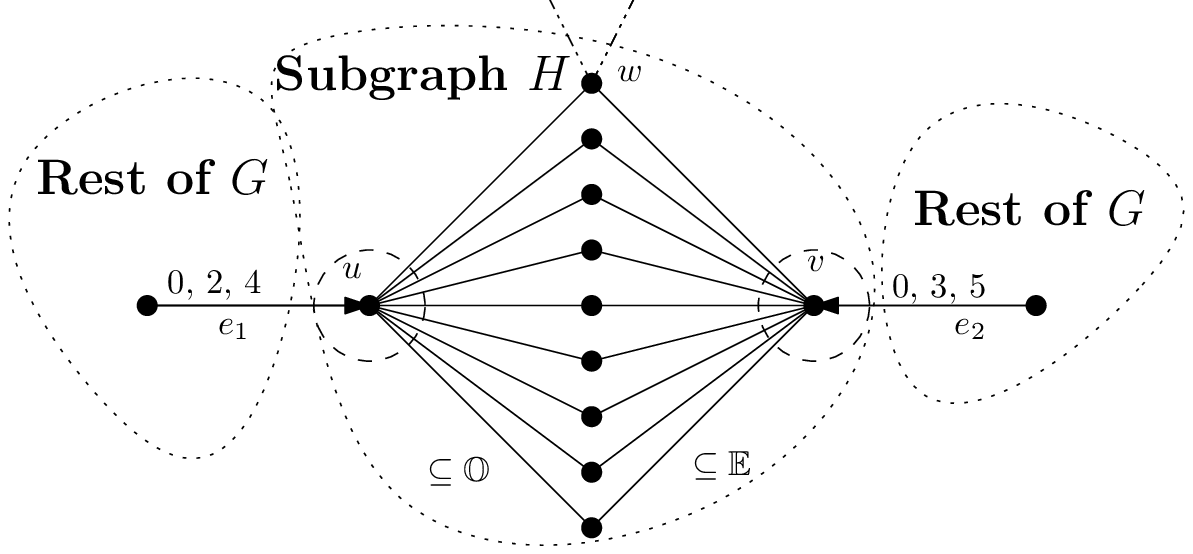}{0.8}

\begin{lemma}[A correct labeling of joint even and odd part]\label{lem:joint}
	Let $\lambda\in\N$, let $G$ be a graph with $\lcdj{G}\leq \lambda$ and $H$ be its subgraph represented by complete bipartite graph $K_{2,\max-1}$ such that:
	\begin{itemize}
		\item  The only two edges connecting $G\setminus H$ to $H$ are $e_1$ and $e_2$, where $u\in e_1$ and $v\in e_2$. 
		\item The graph $H$ contains vertices $u\neq v$, $deg_G(u)=deg_G(v)\geq 4$ and their common neighbors, call them $N$. Vertices from $N$ are not adjacent, but exactly one of them $w$ may have zero, one or two other neighbors outside $H$.
		\item Moreover, each edge $\{u,z\}, z\in N$ can be labeled only by odd labels $(\Ol)$ and each edge $\{v,z\}, z\in N$ can be labeled only by even labels $(\El)$ and has no other condition on them from the rest of $G$. (It's essential that they can be labeled by arbitrary label of appropriate set except the labels of edges $e_1$ and $e_2$.)
	\end{itemize}
	We have four cases which depends on labels of $e_1$ and $e_2$, on the degree of $u$ and $v$ and on the number of neighbors of $w$.
	If one of the following cases happen:
	\begin{itemize}
		\item[I.] Both $e_1,\ e_2$ have label $0$, $deg_G(u)=deg_G(v)=\max$ and the vertex $w$ has one output edge. (for $\lambda$ odd)
		\item[II.] Both $e_1,\ e_2$ have label $0$, $deg_G(u)=deg_G(v)=\max -1$ and vertex $w$ has two output edges. (for $\lambda$ even)

		\item[III.] The edge $e_1$ has label $2$ and edge $e_2$ has label $3$ and $deg_G(v)=deg_G(u)=\max -1$. (for $\lambda$ odd)
		\item[IV.] The edge $e_1$ has label $4$ and edge $e_2$ has label $5$  and $deg_G(v)=deg_G(u)=\max -1$. (for $\lambda$ odd)
	\end{itemize}
	\emph{Then} all edges incident to vertices of $N$ can be labeled correctly.
	\begin{itemize}
		\item[I.] The output edge incident to $w$ has to have a label $1$.
		\item[II.] The output edges incident to $w$ has to have $1$ and some $s\neq 0$ even. 
	\end{itemize}
\end{lemma}

We omit the full proof of this technical lemma here but it is proved in the appendix. The idea of the proof is to construct an auxiliary bipartite graph. Each edge of $H$ is labeled by some label from the correct set and it is represented by a vertex. Two vertices are connected whenever they be incident in graph $H$ without breaking condition of a correct labeling. It can be shown that such graph is almost $k$-regular for some $k$. Moreover we can delete some edges from that graph and then it becomes $k$-regular. Then we can found perfect matching using Hall marriage theorem.

The Labeling of the output edge is then easy to show because label $1$ is the only unused label and it cannot be placed anywhere else. The other edge incident to the vertex $w$ has an arbitrary nonzero even label and we have exactly one left.

\paragraph{The main reductions proof idea}
We would like to give a reader the general idea used in proofs of all cases. We will develop some gadgets to model the two parts of the input of \MNpr. Namely the logical variables and the formula itself, which we model clause by clause. Moreover, in general-case reductions we need some middle-pieces to glue them together. 

To prove that the gadget for a variable works correctly we need to check that there is no any other labeling of output edges in the {\em variable gadget} than the one described in the image, or its $\lambda$-inversion.
Note that the only possible labels on an output edge are $0$ (or $1$) and $\lambda$ (or $\lambda-1$)---these will represent the logical value of the variable.
For now on, we omit the $\lambda$-inversion case in the proof.
Every variable gadget contains a part with an output edge such that it is possible to repeat it arbitrarily---we call this part {\em repeatable}.

For a clause, we use a gadget for a given span with exactly $3$ input edges. This \emph{clause gadget} has to admit a labeling whenever at most two input edges represents the same logical value.
On the other hand it does not admit a labeling when all input edges represents the same logical value.

\subsection{Even $\lambda\ge 8$}\label{sec:even}
We divide the \emph{variable gadget} into three parts. The \emph{initial part} and the \emph{ending part} are only technical support for starting and ending process correctly. The main work is done in the \emph{repeatable part}.

\imgw{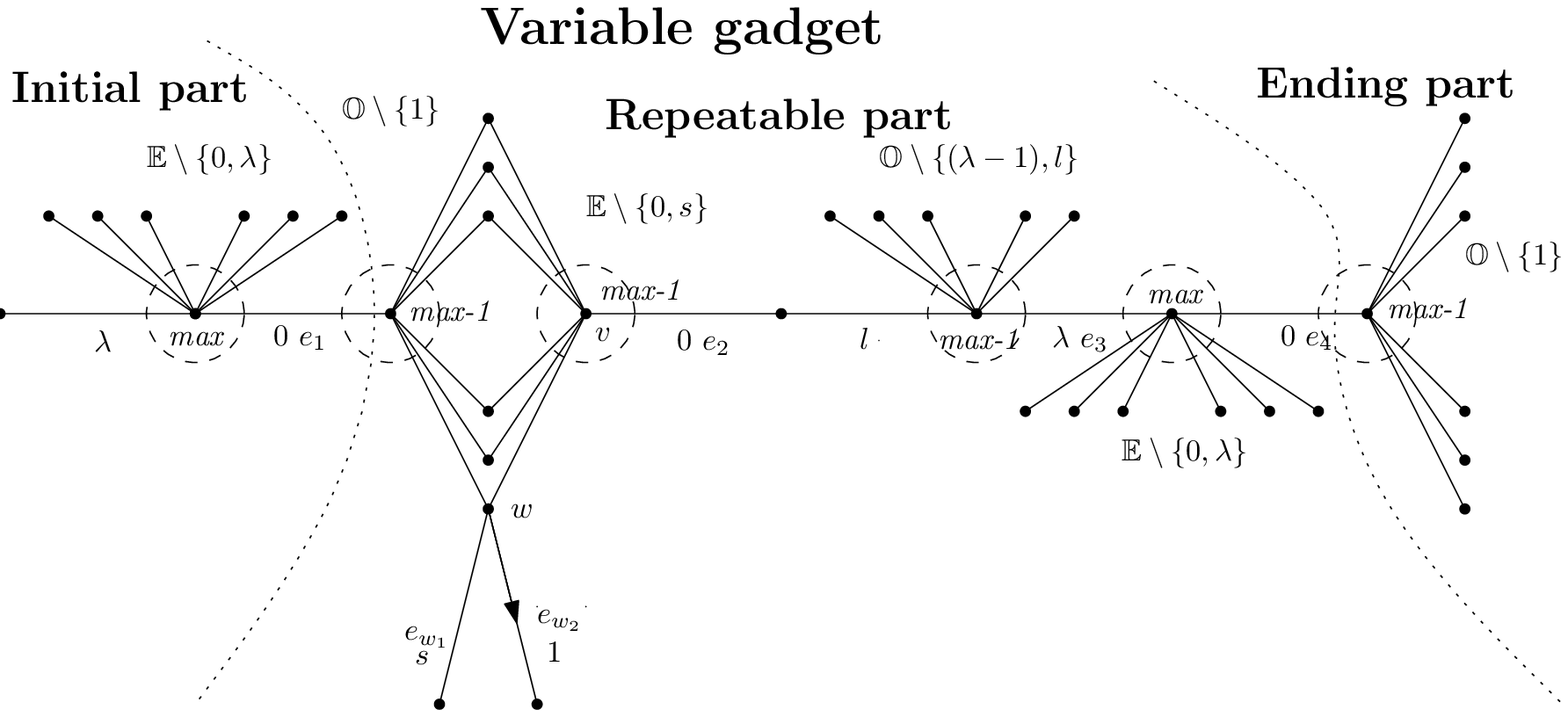}{0.85}

By Lemma~\ref{lem:even} the label of $e_1$ is $0$. Now we have two possibilities (sets) how to label all the edges incident to the vertex $v$: $\El\cup\{1\}\setminus\{0,2\}$ and $\El\setminus\{s\in\El\}.$
If we label edges incident to $v$ from the set $\El\cup\{1\}\setminus\{0,2\}$ it is impossible to label both edges $e_{w_1}$, $e_{w_2}$ incident to the vertex $w$, because we need to use both $0,\ 2$ labels on them. But the label $0$ is already used for the edge $e_1$ which is at distance two.
While if we label these edges from the set $\El\setminus\{s\in\El\}$, in this case it is possible to label the output edge by $s$ or by $1.$

Later the middle-piece gadget further restricts the output, so that the only possible label is $1.$

To prove that it is correct we use Lemma~\ref{lem:joint} part II.

Edges $e_3$ and $e_4$ need to have labels $0$ or $\lambda$ by Lemma~\ref{lem:even}.
As $e_3$ is in distance two to $e_2$ and $e_2$ is labeled by $0$ implies that $e_3$ cannot have label $0.$

The {\em middle-piece gadget} gives us only two possible outputs: $2$ or $0$. This is because Lemma~\ref{lem:max}.
Moreover, this implies that the only possible labeling of input edges is by the label $1.$

The output of the middle-piece gadget is plugged into the \emph{clause gadget}.
\imgw{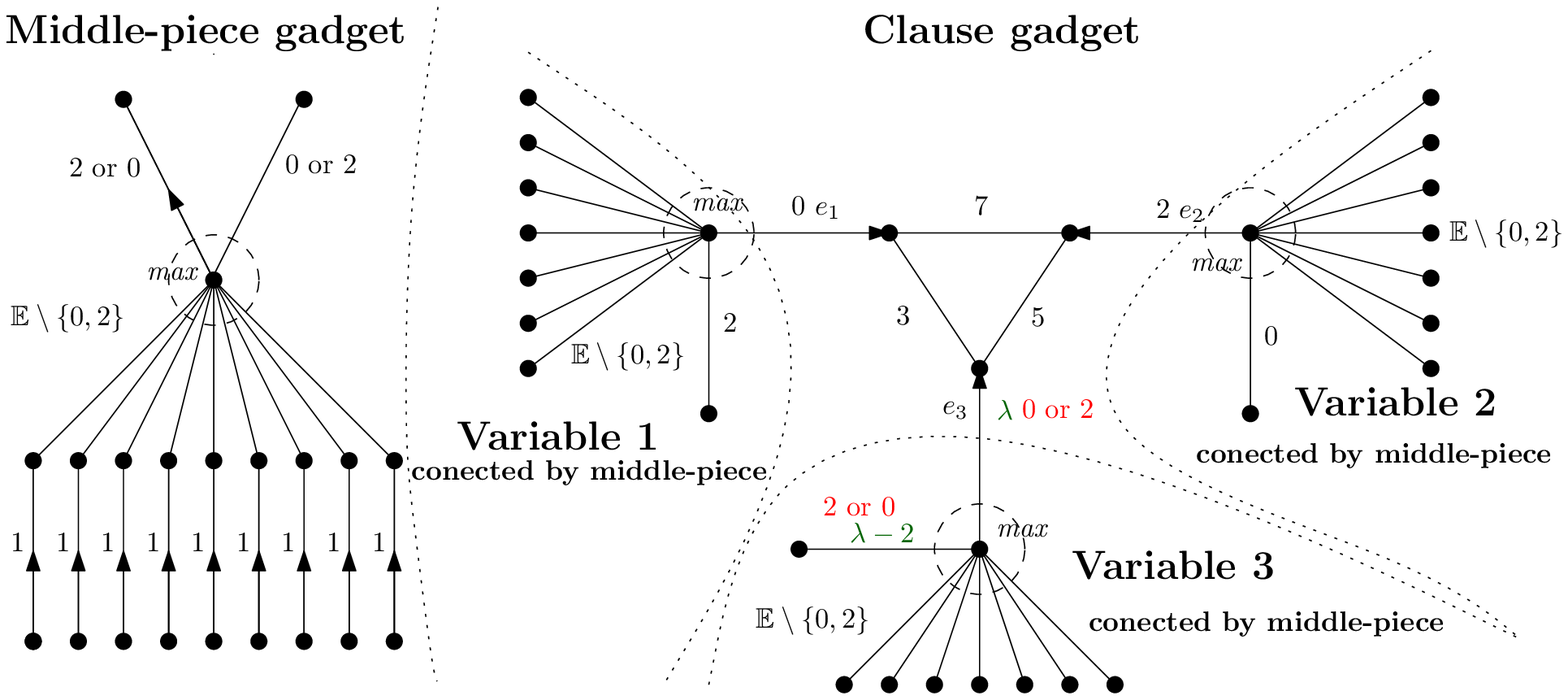}{0.99}\qed


\subsection{Odd $\lambda\ge 9$}\label{sec:odd}
This case is more complicated than the previous one. A reason for this is in the difference between Lemma~\ref{lem:odd} and Lemma~\ref{lem:even}.
In either case there are only two possible labelings, but in Lemma~\ref{lem:odd} the degree of the vertex $u$ equals to the degree of the vertex $v$, while this is not true in Lemma~\ref{lem:even} and so we can distinguish them in the even case shown before.

We start with correctness of the \emph{variable gadget}.
We prove that neighboring edges of vertex $v$ are labeled by labels from $\Ol\setminus\{1\}\cup\{0\}.$ We proceed by contradiction. Suppose that these edges are labeled by $\El$ (according to Lemma~\ref{lem:max} this is the only other option) then edges incident to the vertex $u$ has labels from $\Ol\setminus\{1\}\cup\{0\}$. Then exists the edge $e=\{v,z\}$ that is labeled by some odd $l\neq\lambda$. So the neighborhood of the vertex $z$ can be labeled either by a set $\El\setminus\{0,2,l-1,l+1\}\cup\{1\}$ or by a set $\El\setminus\{0,l-1,l+1\}$. Neither of them is sufficiently large to label all the edges.

\imgw{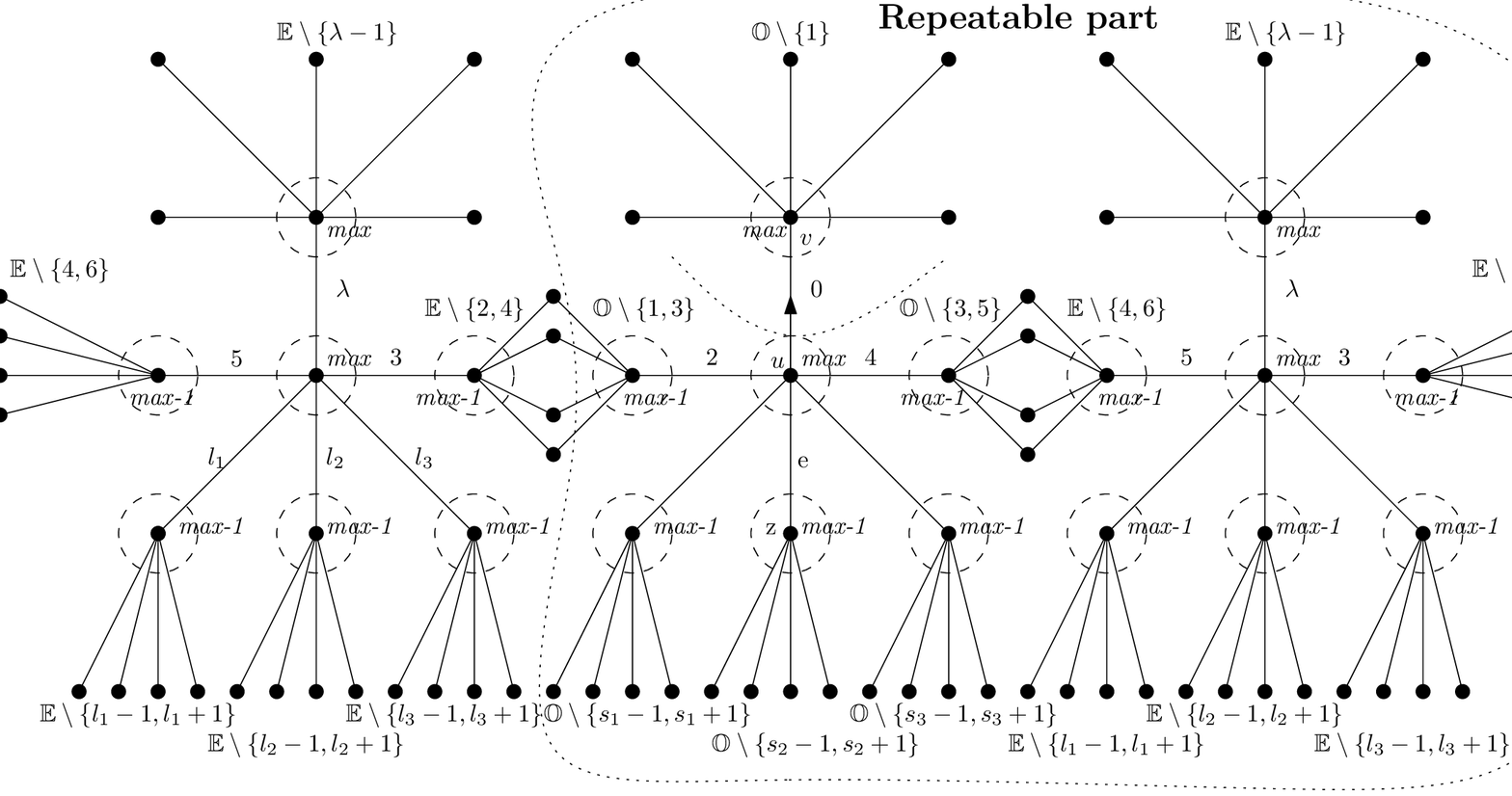}{0.99}

The correctness of the other labeling is shown in the image.

Lemma~\ref{lem:joint} parts III. and IV. ensures that it is possible to repeat the repeatable part of the gadget.
Note that the repeatable part consists of two identical parts, but it is possible to use only one of them as an output, because these parts are labeled $\lambda$-symmetrically.

\imgw{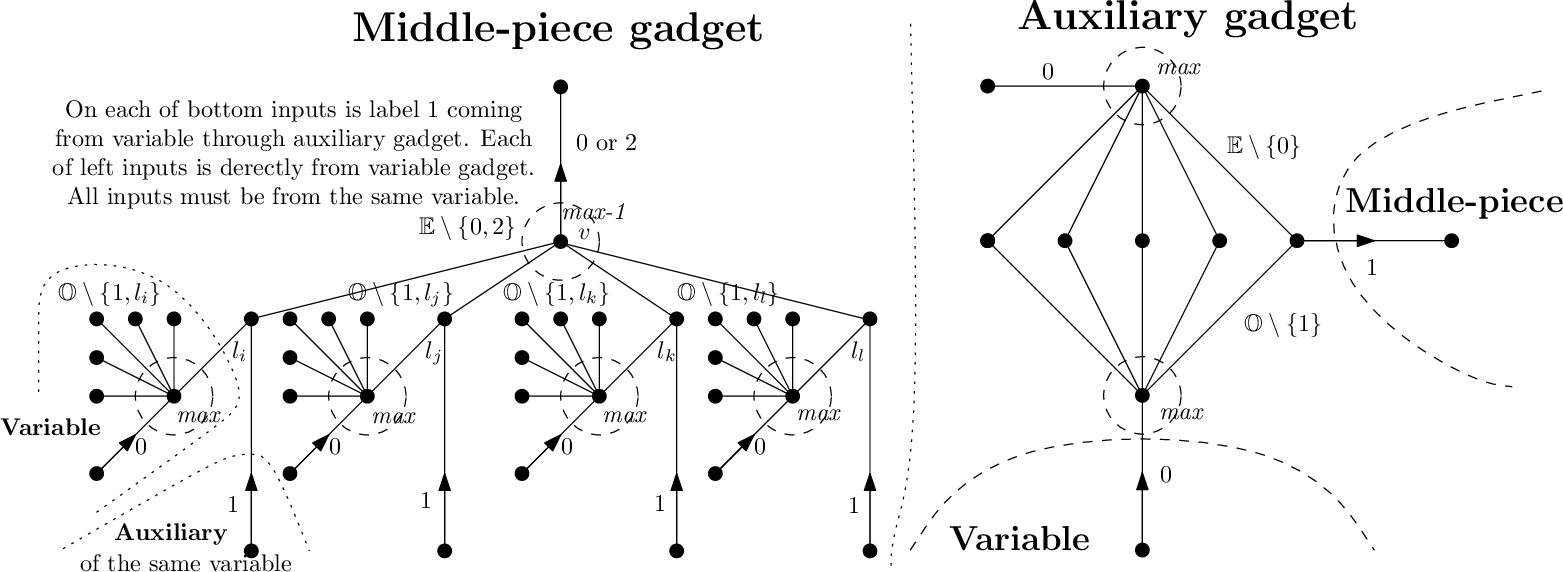}{1}

The correctness of the \emph{auxiliary gadget} is described in Lemma~\ref{lem:joint} part I. The purpose of this gadget is to create an edge with label $1$.

The \emph{middle-piece gadget} has two kinds of inputs. Both kinds of inputs correspond to the \emph{variable gadget}, but one of them is connected to the middle-piece through the \emph{auxiliary gadget}. 

The edges incident to the vertex $v$ can by labeled only by labels from the set $\{\El\}$. This is ensured by the variable inputs, because they contains each label from the set $\Ol\setminus\{1\}$ and also by auxiliary inputs containing label $1$.
Note, that we can create as many such inputs as it is needed.
Moreover, the label $1$ forbids labels $0$ and $2$ anywhere besides the output edge.


Each output from the \emph{middle-piece gadget} is plugged into the \emph{clause gadget} in the following way, which completes the proof.

\img{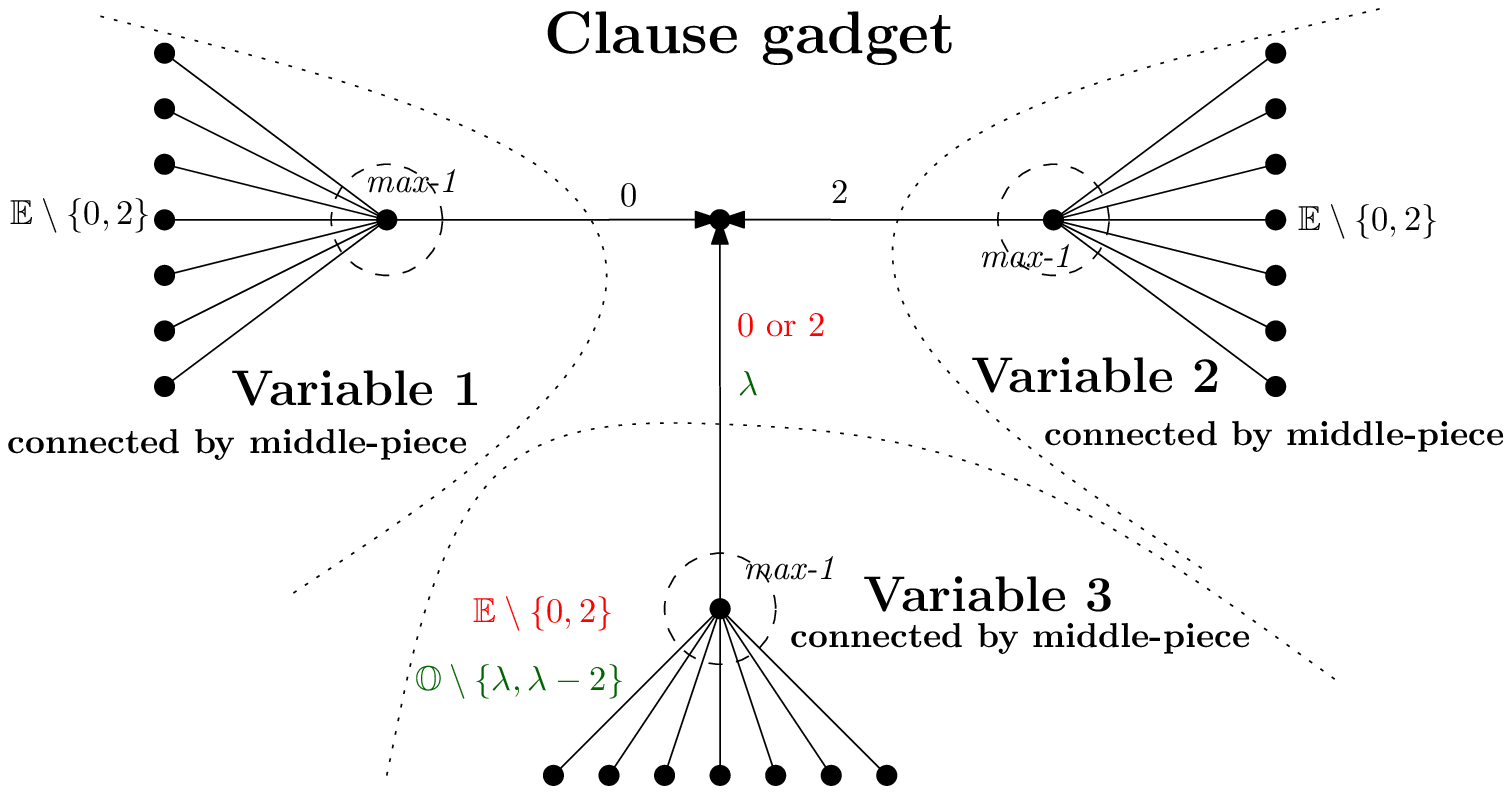}{0.75}\qed

%% file: src/appendix.tex
\clearpage
\section*{Appendix}

\subsection*{Proof of Lemma~\ref{lem:gc}}

\begin{proof}
For the proof we use all the facts already proved in the proof of Lemma~\ref{lem:gp}.
The difference between a generalized path and a cycle is that the generalized cycle is closed, and so we have to care about used labeling.

For all lengths of sequences presented so far, the sequence starts with the label $0$ and ends with the label $4.$ Recall that the sequence of length $10$ was $|0314204130|.$ This sequence start and ends with the same label---and this cause the incorrectness of labeling.

First we present a new sequence $|03140240240314|(d = 14)$---this sequence proves, that the only sequence that has to start and end with the same label is the one with length $10.$
Now it is clear that if there is even number of pairs with $d = 10$ then the constructed labeling is correct, which finishes the proof of the first part.

For the second part of the lemma, we have to show that for $d = 13, 14$ or $d\ge 16$, there is also a sequence that starts and ends with the label $0$ and the impossibility of such a labeling for all other $d.$
As usually, we begin with the desired sequences $|4130240240314|(d=13), |41302403140314|(d=14), |413041302403140314|(d=18).$ In all these sequences the subsequence $024$ can be repeated arbitrarily.

For the rest we already know, that all the sequences that starts and ends with $0$ have to start with the subsequence $|0314$ and end with the subsequence $4130|.$ As these subsequences cannot be glued together, we have to glue them through another subsequence, which we call a {\em connector}.
Note that the connector subsequence cannot be of length $1,$ because the starting and ending subsequences starts and ends with the same label. This already forbids all $d\le 10.$

The connector can be the sequence $20$ or $02.$ The resulting sequence is the sequence $0314204130,$ which we are already familiar with. Again it is impossible to prolong the sequence a subsequence of length one, two or five. It is easy to see that the only possibilities are to put
\begin{itemize}
  \item[(i)] $0314$ to the beginning, 
  \item[(ii)] $4130$ to the end,
  \item[(iii)] $420$ right after the connector.
\end{itemize}
This forbids the sequences of length $11,\ 12$ and $15$ and finishes the proof.
\qed
\end{proof}

\subsection*{Proof of Lemma~\ref{lem:joint}}

More detailed proof of the lemma follows the idea that has been shown before.

For readers convenience, we repeat here the figure corresponding to the Lemma~\ref{lem:joint}:

\imgw{eng-sude_liche}{1}

\begin{proof}
We start by a construction of an auxiliary $2$-regular bipartite graph $H_A.$
We use the graph $H_A$ to represent the incompatibility relation between the set $\El$ and the set $\Ol.$
Recall that $k = \deg(u) = \deg(v)$. The {\em left partite} represents $k-1$ odd labels, by which we can label the $H$-neighborhood of the vertex $u$. While the {\em right partite} represents $k-1$ even labels, by which we can label the $H$-neighborhood of the vertex $v$.
Of course by this we do not use the labels of edges $e_1$ and $e_2$.

	Vertices are connected by an edge, whenever corresponding edges in graph $H$ cannot be incident.

	Note that every vertex in graph $H_A$ has degree at most $2,$ as we would like $H_A$ to be $2$-regular, we have to add several edges to $H_A,$ which we do as follows:
	\begin{itemize}
		\item[I.] In this case the left partite represents labels in the set $\Ol\setminus\{1\},$ while the right partite represents labels in the set $\El\setminus\{0\}$. The only vertices with degree one are: $\lambda$ and $2$.
		It is possible to add an edge $\{2,\lambda\}$.
	\item[II.] The left  partite represents labels in the set $\Ol\setminus\{1\}$, while the right partite has represents labels in the set $\El\setminus\{0,\lambda\}$. Only vertices with degree one are: $\lambda-1$ and $2$. Then we can add edge $\{2,\lambda-1\}$.
		\item[III.] The left partite represents labels in the set $\Ol\setminus\{1,3\}$, while the right partite represents labels in the set $\El\setminus\{2,4\}$. 
			Vertices with degree less than two represents the following labels: $0$ (degree zero), $5$ and $\lambda$ (both degree one).
			Then we can add two edges: $\{0,5\}$ and $\{0,\lambda\}.$
		\item[IV.] The left partite represents labels in the set $\Ol\setminus\{3,5\}$, while the right partite represents labels in the set $\El\setminus\{4,6\}$. 
			Vertices with degree one represents the following labels: $0$, $2$, $7$ and $\lambda.$
		Then we can add two edges: $\{2,7\}$ and $\{0,\lambda\}.$
	\end{itemize}

	\imgw{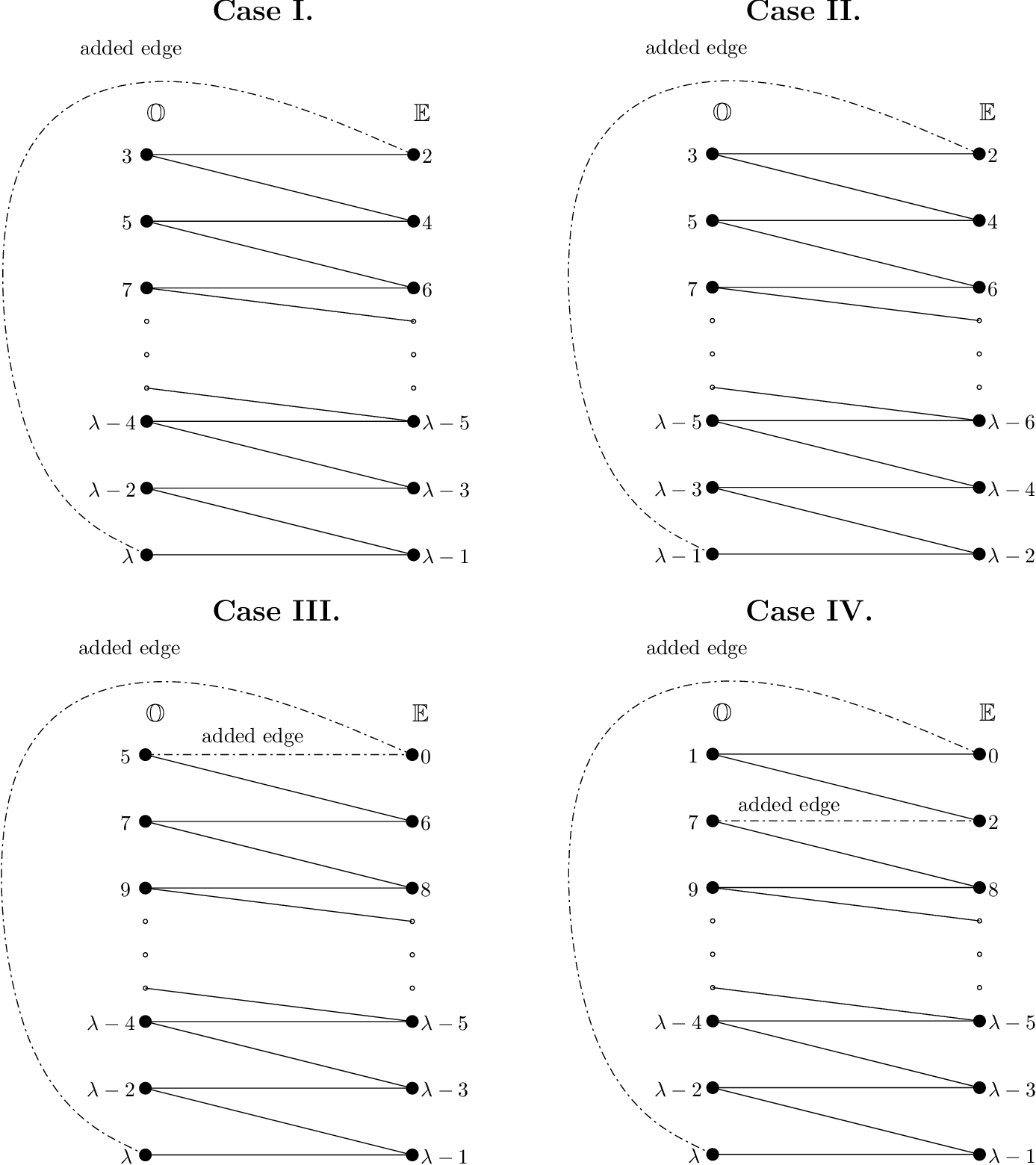}{0.8}

	 Now we create the graph complement $\bar{H_A}$ of the auxiliary graph $H_A$. Now $\bar{H_A}$ is $(k-3)$-regular bipartite graph and then it has perfect matching by Hall's marriage theorem and so this perfect matching describes a correct labeling of the graph $H$.

	Now it remains to show that in cases I. and II. it is possible to extend the labeling to the output edges incident to the vertex $w.$ For this by {\em inner edges incident to $w$} we mean the edges $\{u,w\},\{v,w\}.$
	\begin{itemize}
	\item[I.] In this case the only label incompatible with the label $1$ on the output edge is the label $2,$ but there at least two edges in the matching do not containing the label $2,$ to set labels to the inner edges incident to the vertex $w.$
	\item[II.] In this case we have to label two outgoing edges incident to the vertex $w.$ Note that without loss of generality, we can use labels $1$ and $\lambda.$ As in the previous case, we have to exclude those labelings that associate label $\lambda-1$ or $2$ with an inner edge incident to $w.$ This is possible as there are at least $3$ edges in the perfect matching.\qed
	\end{itemize}
\end{proof}

\subsection*{\NP-hardness for $\lambda = 5$}

\begin{lemma}
The \Lproblem problem is \NP-hard for $\lambda = 5.$
\end{lemma}
\begin{proof}
	A variable of is represented by the following \emph{variable gadget}:
\imgw{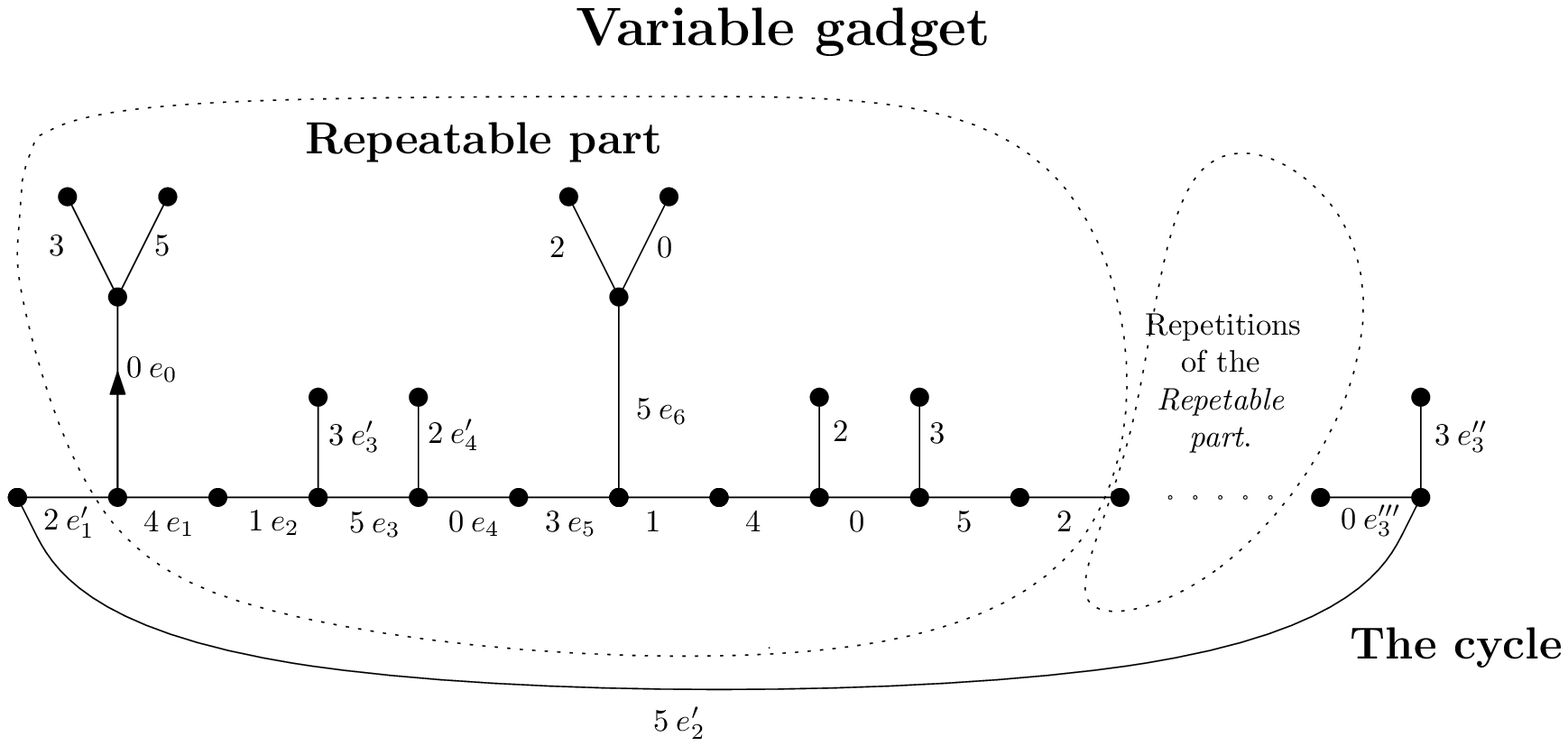}{1}

Case analysis are in tables for now on. In the table is shown every possible labeling of the edges highlighted in gadget starting with the edge $e_0$. If the labeling cannot be extended to all edges it is marked by the symbol "---".

\begin{center}	
	\begin{tabular}{|c||c|c|c|c|c|c|c|}\hline
			&\multirow{5}{*}{$e_0$} & \multirow{2}{*}{$e_1'$}& \multirow{2}{*}{$e_2'$} & $e_3''$ & & & \\
											& &&& $e_3'''$ & & & \\
						  &   &	 \multirow{3}{*}{$e_1$}& \multirow{3}{*}{$e_2$}& $e_3'$& & & \\
				 &		&&	& \multirow{2}{*}{$e_3$}& $e_4'$ & & \\
		   & &&&		& $e_4$ & $e_5$ & $e_6$\\\hline
				\hline
			I.	&\multirow{5}{*}{$0$} & \multirow{2}{*}{$3$}& \multirow{2}{*}{$1$} & $4, 5$ only & & & \\
			   &										 &&& $4, 5$ only  & & & \\
				  &	   &	 \multirow{3}{*}{$5$}& \multirow{3}{*}{$2$}& $4$& & & \\
				  &	&&	& \multirow{2}{*}{$0$}& $3$ & & \\
			&&&	&	& $5$ & more options & ---\\\hline
				II.&\multirow{5}{*}{$0$} & \multirow{2}{*}{$5$}& \multirow{2}{*}{$2$} & $4$ & & & \\
				  &									 &&& $0$ & & & \\
				  &	   &	 \multirow{3}{*}{$3$}& \multirow{3}{*}{$1$}& $4,5$ only & & & \\
				  &	&&	& \multirow{2}{*}{$4, 5$ only}& ---  & & \\
			&&&	&	& --- & --- & ---\\\hline
				III.&	\multirow{5}{*}{$0$} & \multirow{2}{*}{$2$}& \multirow{2}{*}{$5$} & $3$ & & & \\
				  &									 &&& $0$ or $1$. & & & \\
				  &	   &	 \multirow{3}{*}{$4$}& \multirow{3}{*}{$1$}& $5$ & & & \\
				  &	&&	& \multirow{2}{*}{$3$}& only $0$.  & & \\
			&&&	&	& only $0$. & --- & ---\\\hline
				IV.&	\multirow{5}{*}{$0$} & \multirow{2}{*}{$2$}& \multirow{2}{*}{$5$} & $3$ & & & \\
				  &									 &&& $0$ or $1$. & & & \\
				  &	   &	 \multirow{3}{*}{$4$}& \multirow{3}{*}{$1$}& $3$ & & & \\
				  &	&&	& \multirow{2}{*}{$5$}& $0$  & & \\
			&&&	&	& $2$ & $4$ & $0$ impossible\\\hline
	\end{tabular}
\end{center}

	\begin{center}	
		\begin{tabular}{|c||c|c|c|c|c|c|c|}\hline
			&	\multirow{5}{*}{$e_0$} & \multirow{2}{*}{$e_1'$}& \multirow{2}{*}{$e_2'$} & $e_3''$ & & & \\
				 &									 &&& $e_3'''$ & & & \\
				 &	   &	 \multirow{3}{*}{$e_1$}& \multirow{3}{*}{$e_2$}& $e_3'$& & & \\
				 &	&&	& \multirow{2}{*}{$e_3$}& $e_4'$ & & \\
		   &&&	&	& $e_4$ & $e_5$ & $e_6$\\\hline
				\hline
				V.&	\multirow{5}{*}{$0$} & \multirow{2}{*}{$2$}& \multirow{2}{*}{$5$} & $3$ & & & \\
				  &									 &&& $0$ or $1$. & & & \\
				  &	   &	 \multirow{3}{*}{$4$}& \multirow{3}{*}{$1$}& $3$ & & & \\
				  &	&&	& \multirow{2}{*}{$5$}& $2$  & & \\
			&&&	&	& $0$ & $4$ & $0$ impossible \\\hline
				VI.&	\multirow{5}{*}{$0$} & \multirow{2}{*}{$2$}& \multirow{2}{*}{$5$} & $3$ & & & \\
				  &									 &&& $0$ or $1$. & & & \\
				  &	   &	 \multirow{3}{*}{$4$}& \multirow{3}{*}{$1$}& $3$ & & & \\
				  &	&&	& \multirow{2}{*}{$5$}& $2$  & & \\
			&&&	&	& $0$ & $3$ & $5$ \\\hline
			VII.&	\multirow{5}{*}{$0$} & \multirow{2}{*}{$4$}& \multirow{2}{*}{ $1$} & $3$ or $5$. & & & \\
				  &									 &&& $5$ or $3$. & & & \\
				  &	   &	 \multirow{3}{*}{$2$}& \multirow{3}{*}{$5$}& $0$ or $1$ & & & \\
				  &	&&	& \multirow{2}{*}{$3$}& only $1$ or $0$.  & & \\
			&&&	&	& only $1$ or $0$. & --- & ---\\\hline
				VIII.&	\multirow{5}{*}{$0$} & \multirow{2}{*}{$4$}& \multirow{2}{*}{ $1$} & $3$ or $5$. & & & \\
				  &									 &&& $5$ or $3$. & & & \\
				  &	   &	 \multirow{3}{*}{$2$}& \multirow{3}{*}{$5$}& $3$  & & & \\
				  &	&&	& \multirow{2}{*}{$1$}& only 4  & & \\
			&&&	&	& only 4. & --- & ---\\\hline
				IX.&	\multirow{5}{*}{$0$} & \multirow{2}{*}{$4$}& \multirow{2}{*}{ $1$} & $3$ or $5$. & & & \\
				  &									 &&& $5$ or $3$. & & & \\
				  &	   &	 \multirow{3}{*}{$2$}& \multirow{3}{*}{$5$}& $3$  & & & \\
				  &	&&	& \multirow{2}{*}{$0$}& $4$  & & \\
			&&&	&	& $2$ & $5$  & impossible\\\hline
				X.&	\multirow{5}{*}{$0$} & \multirow{2}{*}{$4$}& \multirow{2}{*}{ $1$} & $3$ or $5$. & & & \\
				   &									 &&& $5$ or $3$. & & & \\
					  &   &	 \multirow{3}{*}{$2$}& \multirow{3}{*}{$5$}& $3$  & & & \\
					  &&&	& \multirow{2}{*}{$0$}& $2$  & & \\
			 &&&&		& $4$ & $1$ & $5$\\\hline
			\end{tabular}
	\end{center}

	The inputs of the variable gadgets are plugged into the following \emph{clause gadget}:

\imgw{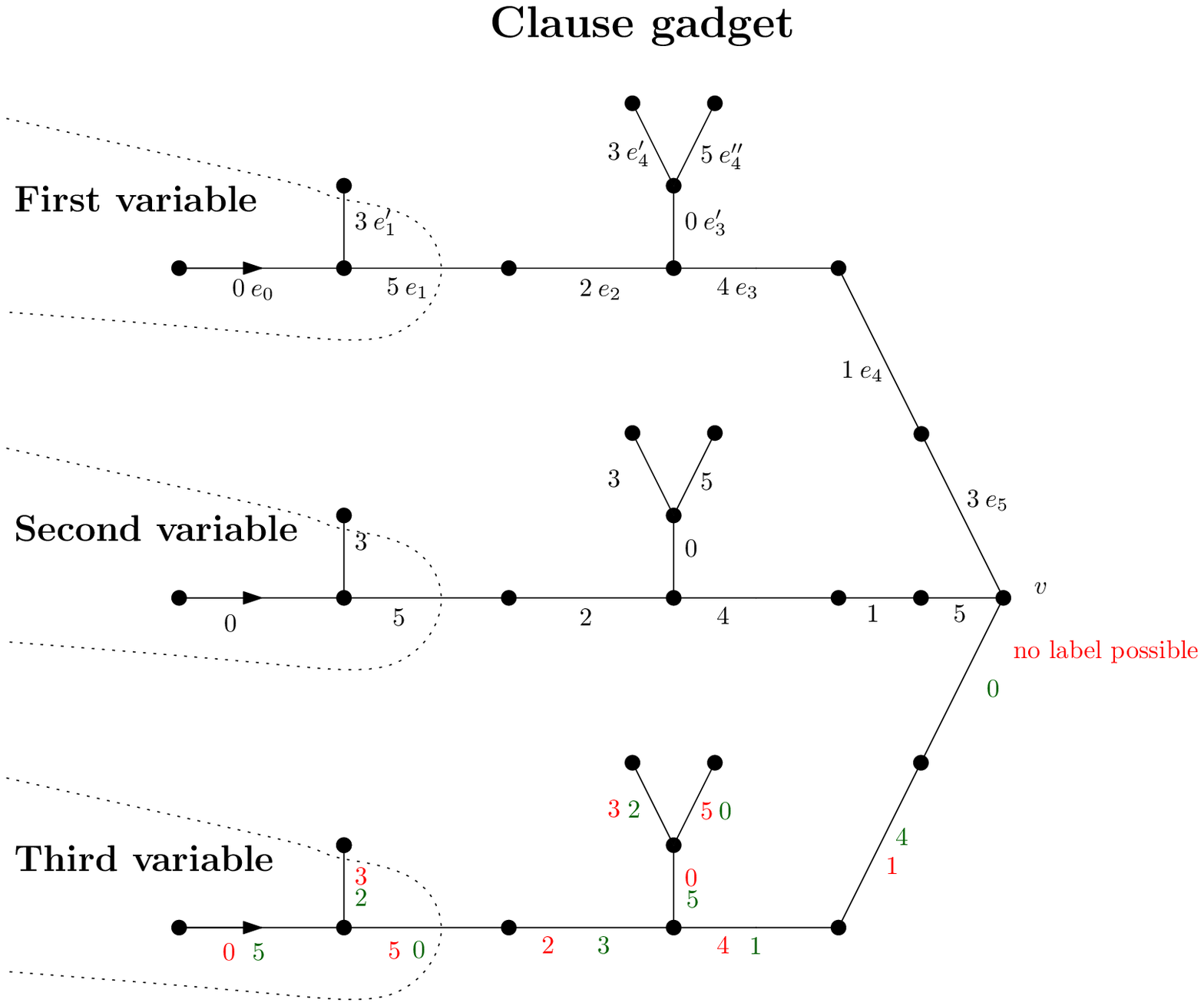}{0.9}
	\begin{center}\label{tabklau521}
		\begin{tabular}{|c||c|c|c|c|c|c|}\hline
			&\multirow{4}{*}{$e_0$} & $e_1'$& & & & \\
			&	&\multirow{3}{*}{$e_1$}& \multirow{3}{*}{$e_2$}& \multirow{2}{*}{$e_3'$}&$e_4'$& \\
			&	&&&&$e_4''$& \\
			&	&&&$e_3$&$e_4$&$e_5$\\\hline 
			\hline
			I.&	\multirow{4}{*}{$0$} & $5$& & & & \\
			  &														 &\multirow{3}{*}{$3$}& \multirow{3}{*}{$1$}& \multirow{2}{*}{$3, 5$ impossible}&---& \\
			  &						&&&&---& \\
			  &	 &&& $3,5$ impossible&---&---\\\hline
			II.&	\multirow{4}{*}{$0$} & $3$& & & & \\
				  &													 &\multirow{3}{*}{$5$}& \multirow{3}{*}{$1$}& \multirow{2}{*}{$3, 5$ impossible}&---& \\
				  &					&&&&---& \\
				  & &&&$3, 5$ impossible&---&---\\\hline
			III.&	\multirow{4}{*}{$0$} & $3$& & & & \\
				  &													 &\multirow{3}{*}{$5$}& \multirow{3}{*}{$2$}& \multirow{2}{*}{$4$}&$0,2$ impossible& \\
				  &					&&&& $0,2$ impossible & \\
				  & &&&$0$&---&---\\\hline
			IV.&	\multirow{4}{*}{$0$} & $3$& & & & \\
				   &													 &\multirow{3}{*}{$5$}& \multirow{3}{*}{$2$}& \multirow{2}{*}{$0$}&$3$ or $5$& \\
					  &				&&&&$5$ or $3$& \\
				   & &&&$4$&$1$&$3$ or $5$\\\hline
		\end{tabular}
	\end{center}


As the only input to the clause gadget can be either from a set $\{3,5\}$ or from a set $\{0,2\}$, which represent the truth assignment of the appropriate variable. From the labels in the gadget, we can see (up to $\lambda$-symmetry) that it is impossible to label the clause if there are three inputs from the set $\{3,5\}$ and it is possible to label the clause if there is at least one input from the other set as it is shown in the image above. \qed

\end{proof}

\subsection*{\NP-hardness for $\lambda = 6$}
\begin{lemma}
The \Lproblem problem is \NP-hard for $\lambda = 6.$
\end{lemma}
A variable is represented by the following \emph{variable gadget}:
\imgw{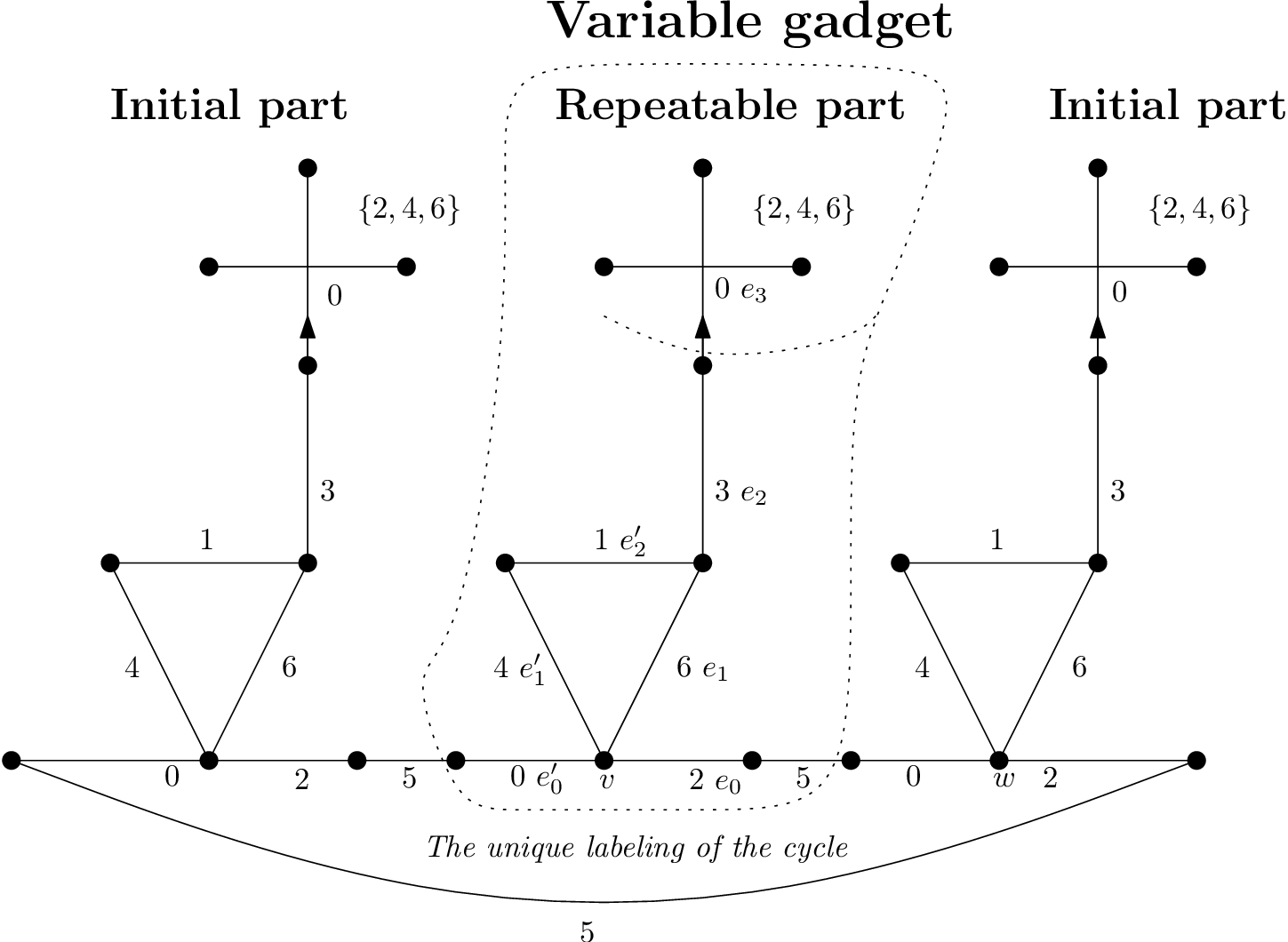}{1}

And the case analysis is in the following table.

	\begin{center}
		\begin{tabular}{|c||c|c|c|c|}\hline
			&	$e_0'$&$e_1'$&$e_2'$ & \\
		 &	$e_0$&$e_1$&$e_2$ &$e_3$ \\
			\hline\hline
		I.	&	$2$&$0$&$3$ & \\
		   &	$4$&$6$&$1$ &$4,6$ impossible in the cycle\\\hline
			II.&	$2$&$6$&$3$ & \\
			   &	$4$&$0$&$5$ & $0,2$ impossible in the cycle\\\hline
		III.&	$0$&$2$&impossible& \\
			&	$6$&$4$&---&--- \\\hline
			IV.&$0$&$2$&impossible& \\
			   &$4$&$6$&---&--- \\\hline
			V.&	$0$&$6$&$1$ & \\
		   &	$2$&$4$&impossible&--- \\\hline
			VI.&	$0$&$4$&$1$ & \\
			&	$2$&$6$&$3$ &$0$ \\\hline
		\end{tabular}
	\end{center}

	And the clause is represented by the \emph{clause gadget}:

\imgw{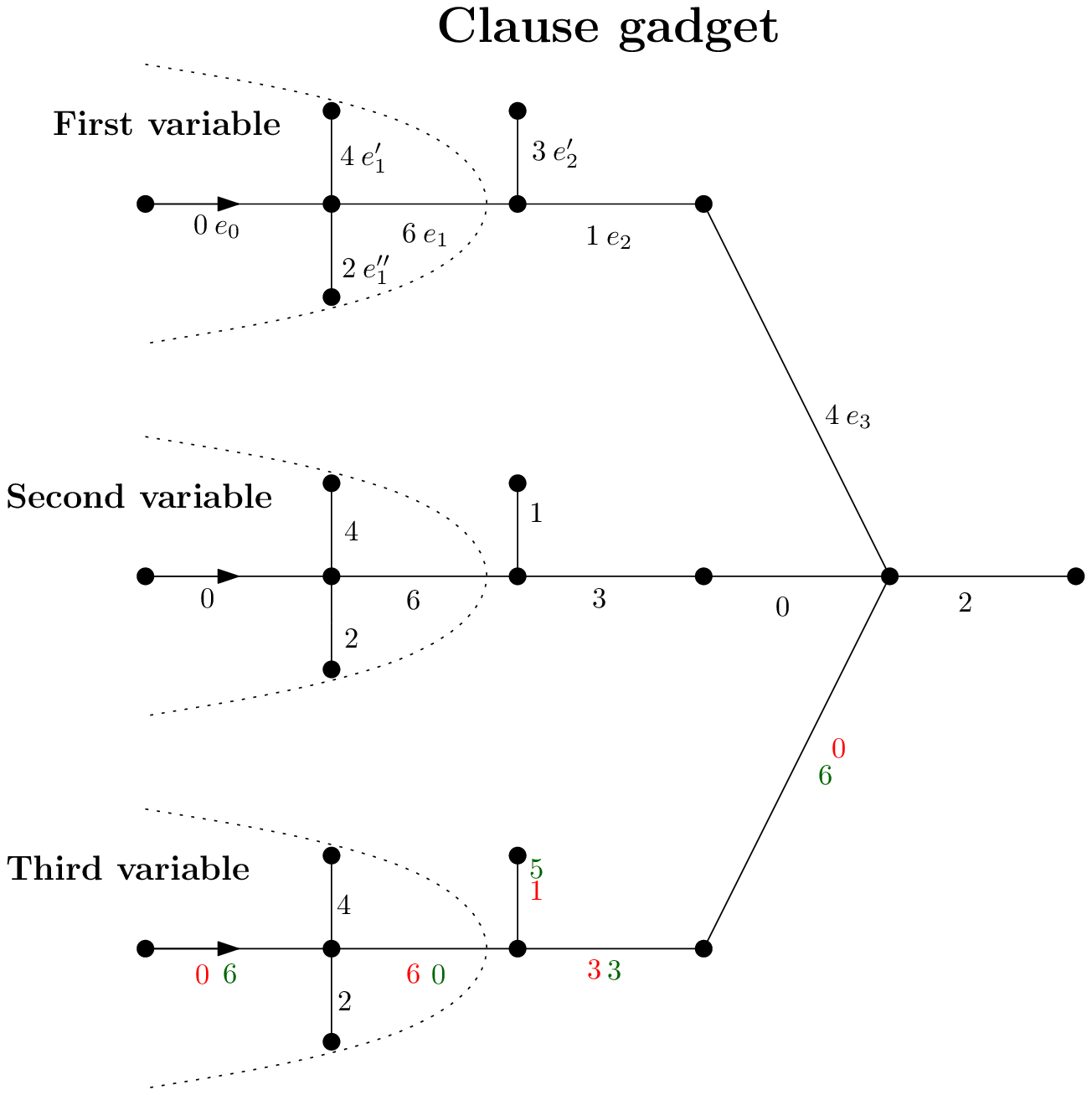}{0.8}

	\begin{center}\label{tabklau621}
		\begin{tabular}{|c||c|c|c|c|}\hline
			&	\multirow{4}{*}{$e_0$}&$e_1''$& & \\
			&												 &$e_1'$& & \\
			&				&\multirow{2}{*}{$e_1$}&$e_2'$& \\
		 & &&$e_2$&$e_3$\\\hline 
			\hline
			I.&	\multirow{4}{*}{$0$}&$4$ or $6$& & \\
			  &											&$6$ or $4$& & \\
			  &			 &\multirow{2}{*}{$2$}&does not have two neighbors with odd label& \\
			  & &&does not have two neighbors with odd label&---\\\hline 
			II.&	\multirow{4}{*}{$0$}&$2$ or $6$& & \\
				  &										&$6$ or $2$& & \\
				  &		 &\multirow{2}{*}{$4$}&does not have two neighbors with odd label& \\
			   & &&does not have two neighbors with odd label&---\\\hline 
			III.&	\multirow{4}{*}{$0$}&$4$ or $2$& & \\
				   &										&$2$ or $4$& & \\
					  &	 &\multirow{2}{*}{$6$}&$1$ or $3$& \\
				& &&$3$ or $1$&($0$, $5$) or ($4$, $5$)\\\hline 
		\end{tabular}
	\end{center}\qed

\subsection*{\NP-hardness variable gadget for $\lambda = 7$}

\begin{lemma}
The \Lproblem problem is \NP-hard for $\lambda = 7.$
\end{lemma}
For the case $\lambda = 7,$ we show only the variable gadget because in this case it is possible to reuse all the other gadgets from the general case where $\lambda\geq 9$.

\imgw{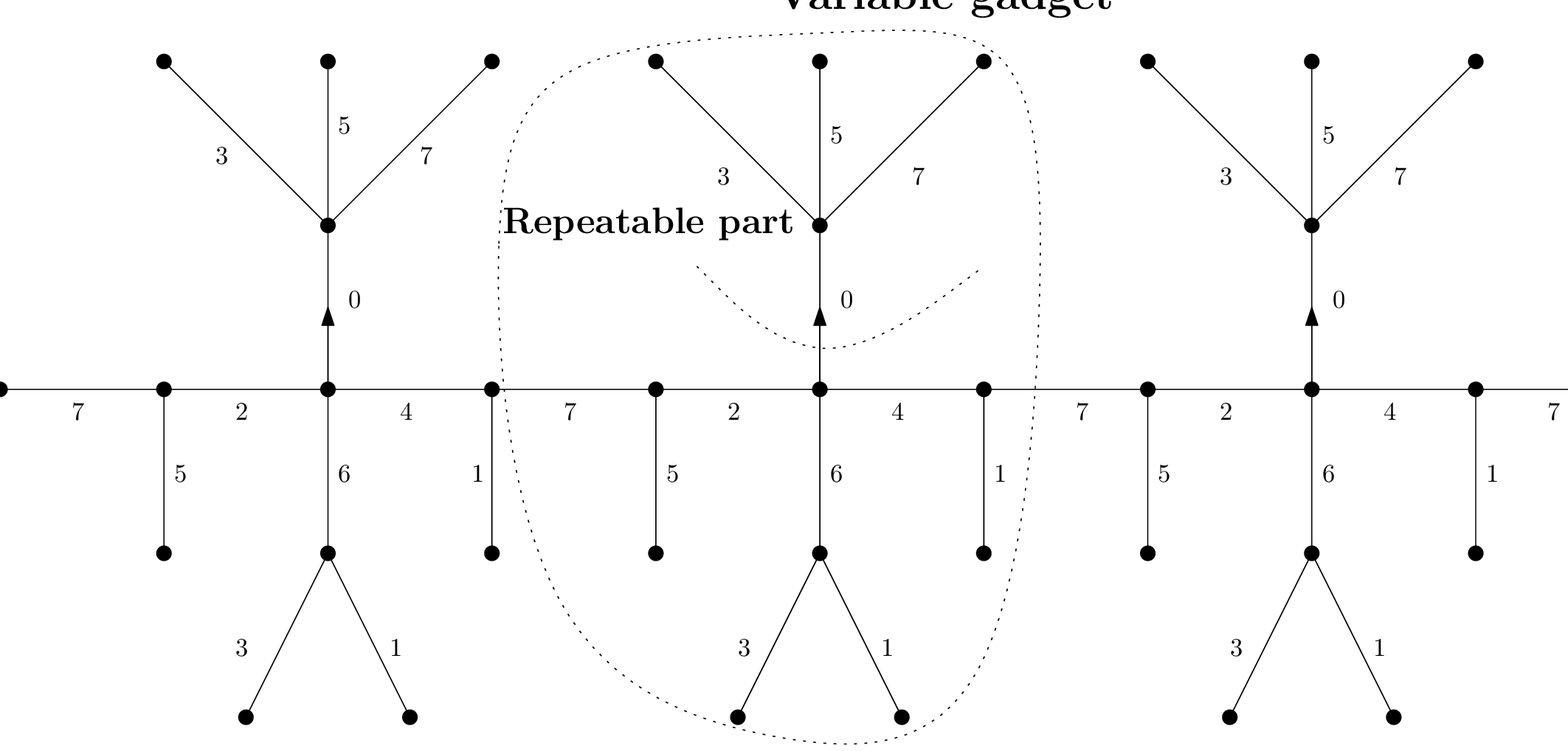}{1}

The correctness of the repeatable part is done by the same argument as it is done in the proof of the general case for $\lambda\geq 9$. Then it is easy to show that the only possible labeling of connection of repeatable parts is the one shown in the image above.\qed